\documentclass{article}
\usepackage{amssymb,amsmath}
\usepackage{graphicx,enumerate}

\title{Expected utility operators and coinsurance problem}

\author{Irina Georgescu \\ \footnotesize Academy of Economic Studies\\ \footnotesize Department of Economic Cybernetics\\ \footnotesize Pia$\c{t}$a Romana No 6  R 70167, Oficiul Postal 22, Bucharest, Romania\\
 \footnotesize Email: irina.georgescu@csie.ase.ro }

\date{}

\begin{document}
\maketitle

\begin{abstract}

The expected utility operators introduced in a previous paper, offer a framework for a general risk aversion theory, in which risk is modelled by a fuzzy number $A$.

In this paper we formulate a coinsurance problem in the possibilistic setting defined by an expected utility operator $T$. Some properties of the optimal saving $T$-coinsurance rate are proved and an approximate calculation formula of this is established with respect to the Arrow-Pratt index of the utility function of the policyholder, as well as the expected value and the variance of a fuzzy number $A$. Various formulas of the optimal $T$-coinsurance rate are deduced for a few expected utility operators in case of a triangular fuzzy number and of some HARA and CRRA-type utility functions.

\end{abstract}

\textbf{Keywords}: expected utility operators, coinsurance

\newtheorem{definitie}{Definition}[section]
\newtheorem{propozitie}[definitie]{Proposition}
\newtheorem{remarca}[definitie]{Remark}
\newtheorem{exemplu}[definitie]{Example}
\newtheorem{intrebare}[definitie]{Open question}
\newtheorem{lema}[definitie]{Lemma}
\newtheorem{teorema}[definitie]{Theorem}
\newtheorem{corolar}[definitie]{Corollary}

\newenvironment{proof}{\noindent\textbf{Proof.}}{\hfill\rule{2mm}{2mm}\vspace*{5mm}}

% This preserves the distinction between vectors and scalars. However,
% if the conference you are submitting to favors bold math in the abstract,
% then you can use LaTeX's standard command \boldmath at the very start
% of the abstract to achieve this. Many IEEE journals/conferences frown on
% math in the abstract anyway.

% no keywords

% For peer review papers, you can put extra information on the cover
% page as needed:
% \ifCLASSOPTIONpeerreview
% \begin{center} \bfseries EDICS Category: 3-BBND \end{center}
% \fi
%
% For peerreview papers, this IEEEtran command inserts a page break and
% creates the second title. It will be ignored for other modes.

\section{Introduction}

In most cases, economic and financial activities are accompanied by risk, which generates pecuniary losses for the agents. A risk-averse agent will try to diminish losses caused by risk by closing an insurance contract. By \cite{eeckhoudt}, in the component of an insurance contract enter a premium $P$ paid by the agent (policyholder) to an insurer and a real function $I(.)$ which specifies the part of the loss that is recovered: if the loss has the size $x$, then the insurer will pay the agent the amount $I(x)$.

Usually, the function $I(.)$ is defined by setting a {\emph{coinsurance rate}} $\beta$, if $x$ is the loss then the policyholder will receive the amount $I(x)=\beta x$. The agent will choose that $\beta$ maximizing the expected utility of her final wealth. So an optimal problem occurs, called the {\emph{coinsurance problem}}.

A probabilistic model of risk assumes that this is mathematically represented by a random variable. The coinsurance problem from \cite{eeckhoudt}, \cite{mossin} is such a probabilistic model, in which the loss is a random variable, and the agent is described by a utility function.

In the paper \cite{georgescu5} there is studied a probabilistic type coinsurance problem: the agent is represented by a utility function, but the loss caused by risk is a fuzzy number.

The coinsurance problem from \cite{georgescu5} is formulated by maximizing a possibilistic expected utility theory, defined in \cite{georgescu2}, \cite{georgescu3} in the framework of a possibilistic treatment of risk aversion. On the other hand, in \cite{georgescu2} there is a second notion of possibilistic expected utility, and the expected utility operators from \cite{georgescu4} allow the definition of a general notion of possibilistic expected utility which generalizes the two mentioned above. By this, to each expected utility operator $T$ one associates a possibilistic expected utility theory (called $T$-possibilistic $EU$-theory), in which various topics on risk theory can be discussed.

The aim of this paper is to study the coinsurance problem in the framework offered by expected utility operators. The formulation of an abstract  coinsurance problem can be done inside an arbitrary  $T$-possibilistic $EU$-theory, but the proofs of the optimal coinsurance rate and its approximate computation assume $T$ to fulfill a supplementary property. For this purpose, the $D$-operators have been chosen, a class of expected utility operators
introduced in \cite{georgescu6} by a preservation condition of derivability of the utility function with respect to a parameter.

Section $2$ presents two notions of possibilistic expected utility from \cite{georgescu1}, \cite{georgescu2}, as well as some operators associated with fuzzy numbers (possibilistic expected utility, possibilistic variance).

In Section $3$ the definition of expected utility operators is recalled from \cite{georgescu3}, \cite{georgescu4} and the $D$-operators are introduced.

To formulate the coinsurance problem in the context of a $T$-possibilistic $EU$-theory, in Section $4$ a few basic notions are defined: the coinsurance contract, the $T$-premium for insurance indemnity, the $T$-coinsurance rate, etc. Assuming that $T$ is a $D$-operator, the $T$-coinsurance rate can be computed as a solution of a first order condition.

The results on the optimal $T$-coinsurance $\beta^\ast$ are contained in Section $5$. A first result is a possibilistic version of a Mossin theorem (\cite{mossin} or \cite{eeckhoudt}, p. 51). The main result is an approximate calculation formula of $\beta^\ast$, with respect to the expected value $E_f(A)$ of the fuzzy number $A$ representing the risk, the $T$-variance $Var_T(A)$ of $A$ and the Arrow-Pratt index of the agent's utility function.
It is also demonstrated a formula that approximates the maximal total expected utility (obtained by the choice of the $T$-coinsurance rate $\beta^\ast$).

Another result of the section asserts that if an agent $u_1$ ia more risk averse than an agent $u_2$, then the optimal $T$-coinsurance rate is higher for $u_1$ than for $u_2$.

In Section $6$ forms of the approximation formula of the optimal $T$-coinsurance rate in the particular case of the expected utility operators $T_1$, $T_2$ from \cite{georgescu1}, \cite{georgescu2} and a risk represented by a triangular fuzzy number are obtained. These formulas are applied for HARA and CRRA-type utility functions.

In the concluding remarks section a few open issues are commented and the result from Appendix presents a necessary condition for the positivity of the optimal $T$-coinsurance rate.

\section{Indicators of fuzzy numbers}

Fuzzy numbers are generalizations of real numbers, able to express imprecise information. Using Zadeh's extension principle \cite{zadeh2}, the operations with real numbers can be extended to fuzzy numbers, such most of algebraic properties are preserved \cite{dubois1}, \cite{dubois2}. At the same time fuzzy numbers can be thought of as possibilistic distributions \cite{dubois2}, \cite{carlsson2}. By parallelism with probabilistic distributions, but in a completely different way, with each fuzzy number possibilistic indicators can be associated: expected value, variance, covariance, moments, etc. \cite{carlsson1}, \cite{carlsson4}, \cite{dubois3}, \cite{fuller2}, \cite{georgescu1}, \cite{zhang}.
All these make the fuzzy numbers an effective tool in the possibilistic treatment of some topics on risk theory \cite{carlsson2}, \cite{collan}, \cite{georgescu3}, \cite{thavaneswaran}, \cite{zhang}.

In this section we will present after \cite{georgescu1}, \cite{georgescu2}, \cite{georgescu3} two notions of expected utility associated with a triple consisting of a utility function (representing an agent), a fuzzy number (representing the risk) and a weighting function. Also, we will recall the definition of expected value and two variances associated with a fuzzy number \cite{carlsson1}, \cite{carlsson2}, \cite{georgescu3}, \cite{zhang}.

We fix a mathematical framework consisting of three entities:

$\bullet$ a weighting function $f: [0, 1] \rightarrow {\mathbb{R}}$ ($f$ is a non-negative and increasing function that satisfies $\int_0^1 f(\gamma) d\gamma =1$);

$\bullet $ a utility function $u: {\mathbb{R}} \rightarrow {\mathbb{R}}$ of class $C^2$;

$\bullet$ a fuzzy number $A$ whose level sets have the form $[A]^\gamma=[a_1(\gamma), a_2(\gamma)]$ for all $\gamma \in [0, 1]$.

Then the support of a fuzzy number $A$ will be $supp(A)=\{x \in {\bf{R}}|A(x)>0\}=(a_1(0), a_2(0))$.

Following \cite{georgescu3} Section $4.2$ we will define the following two notions of possibilistic expected utility:

$E_1(f, u(A))=\frac{1}{2} \int_0^1 [u(a_1(\gamma))+u(a_2(\gamma))]f(\gamma) d\gamma$ (2.1)

$E_2(f, u(A))=\frac{1}{2} \int_0^1 [\frac{1}{a_2(\gamma)-a_1(\gamma)} \int_{a_1(\gamma)}^{a_2(\gamma)} u(x)dx]f(\gamma) d\gamma$ (2.2)

Setting in (2.1) or (2.2) $u=1_{\bf{R}}$ (the identity function of ${\mathbb{R}}$) one obtains the possibilistic expected value (\cite{dubois3}, \cite{carlsson1}, \cite{fuller2}, \cite{majlender}):

$E_f(A)=E_1(f, 1_{\bf{R}}(A))=E_2(f, 1_{\bf{R}}(A))=\frac{1}{2}\int_0^1 [a_1(\gamma)+a_2(\gamma)]f(\gamma)d\gamma$ (2.3)

If $supp(A) \subseteq {\bf{R}}_+$, then from (2.3) it follows that $E_f(A)\geq 0$.

For $u(x)=(x-E_f(A))^2$  two different notions of possibilistic variance follow \cite{carlsson1}, \cite{fuller2}, \cite{georgescu1}, \cite{zhang}:

$Var_1(f, A)=\frac{1}{2} \int_0^1 [(a_1(\gamma)-E_f(A))^2+(a_2(\gamma)-E_f(A))^2] f(\gamma) d\gamma$ (2.4)

$Var_2(f, A)=\int_0^1 [\frac{1}{a_2(\gamma)-a_1(\gamma)} \int_{a_1(\gamma)}^{a_2(\gamma)} (x-E_f(A))^2] f(\gamma)d\gamma$ (2.5)

\section{Expected utility operators and $D$-operators}

In this section we will recall the definitions of the expected utility operators and the $D$-operators, introduced in \cite{georgescu2}, respectively \cite{georgescu6} as an abstraction of the two possibilistic expected utilities $E_1(f, u(A))$ and $E_2(f, u(A))$ from the previous section.

Let $\mathcal{F}$ be the set of fuzzy numbers, $\mathcal{C}({\mathbb{R}})$ the set of real continuous functions (mapped from $\mathbb{R}$  to $\mathbb{R}$) and $\mathcal{U}$ a subset of $\mathcal{C}({\mathbb{R}})$ satisfying the following properties:

$(U_1)$ $\mathcal{U}$ contains constant functions and first and second degree polynomial functions;

$(U_2)$ $\mathcal{U}$ is closed under linear combinations: if $a, b \in {\mathbb{R}}$ and $g, h \in \mathcal{U}$ then $ag+bh \in \mathcal{U}$.

For each $a \in {\mathbb{R}}$ we denote $\bar a : {\mathbb{R}} \rightarrow {\mathbb{R}}$ the constant function $\bar a(x)=a$, for $x \in {\mathbb{R}}$. ${\bf{1_{R}}}$ will be the identity function of ${\mathbb{R}}$. Then $\bar a,{\bf{1_{R}}}$  belong to $\mathcal{U}$. In particular, we can consider $\mathcal{U}=\mathcal{C}({\mathbb{R}})$.

We fix a weighting function $f: [0, 1] \rightarrow {\mathbb{R}}$ and a family $\mathcal{U}$ with the properties $(U_1)$ and $(U_2)$.

\begin{definitie}
\cite{georgescu3}, \cite{georgescu4} An {\emph{(f-weighted) expected utility operator}} is a function $T: \mathcal{F} \times \mathcal{U} \rightarrow \mathbb{R}$ such that for any $a, b \in {\mathbb{R}}$, $g, h \in \mathcal{U}$ and $A \in \mathcal{F}$ the following conditions are fulfilled:

(a) $T(A, {\bf{1_{R}}})=E_f(A)$;

(b) $T(A, \bar a)=a$;

(c) $T(A, ag+bh)=aT(A, g)+bT(A, h)$;

(d) $g \leq h$ implies $T(A, g) \leq T(A, h)$.
\end{definitie}

\begin{exemplu}
\cite{georgescu1}, \cite{georgescu3} We consider the function $T_1: \mathcal{F} \times \mathcal{C}({\mathbb{R}}) \rightarrow \mathbb{R}$ defined as follows: for any fuzzy number $A$ and for any $g \in \mathcal{C(\mathbb{R})}$:

$T_1(A, g)=E_1(f, g(A))$  (3.1)

Then $T_1$ is an expected utility operator.
\end{exemplu}

\begin{exemplu}
\cite{georgescu1}, \cite{georgescu3} We consider the function $T_2: \mathcal{F} \times \mathcal{C}({\mathbb{R}}) \rightarrow \mathbb{R}$ defined as follows: for any fuzzy number $A$ and for any $g \in \mathcal{C(\mathbb{R})}$:

$T_2(A, g)=E_2(f, g(A))$  (3.2)

Then $T_2$ is an expected utility operator.
\end{exemplu}

An expected utility operator $T$ is strictly increasing if for any $A \in \mathcal{F}$ and $g, h \in {\mathcal{U}}$, $g<h$ implies $T(A, g)<T(A, h)$. One can prove that the expected utility operators $T_1, T_2$ are strictly increasing.

If $T$ is a strictly increasing operator, then for any $A \in \mathcal{F}$ and $h \in \mathcal{U}$, $h>0$ implies $T(A, h)>0$.

One notices that the axioms (a)-(d) of Definition 3.1 have been abstracted from the properties of $T_1$ and $T_2$. Therefore, the real number $T(A, g)$ will be called {\emph{generalized possibilistic expected utility}} (shortly, $T$-expected utility) and it will represent the starting point of a possibilistic  $EU$-theory associated with $T$. Sometimes, instead of $T(A, g)$ we will use the notation $T(A, g(x))$.

Particularizing $g$, from $T(A, g)$ various possibilistic indicators are obtained. By axiom (a), for $g={\bf{1_{R}}}$, the possibilistic expected value $E_f(A)$ follows. For $g(x)=(x-E_f(A))^2$ we have the notion of $T$-covariance:

$Var_T(A)=T(A, (x-E_f(A))^2)$ (3.3)

Using the axiom (d) of Definition 3.1, it follows immediately that $Var_T(A)\geq 0$.

\begin{remarca}
As we have seen previously, the two operators $T_1$ and $T_2$ introduce the possibilistic variances $Var_{T_1}(A)=Var_1(f, A)$, respectively $Var_{T_2}(A)=Var_2(f, A)$.
\end{remarca}

The two possibilistic variances $Var_{T_1}(A)$, $Var_{T_2}(A)$ have been used in the application of some models in possibilistic risk theory \cite{appadoo}, \cite{carlsson2}, \cite{collan}, \cite{georgescu5}, \cite{majlender}, \cite{thavaneswaran}.

In case of probabilistic risk, the risk aversion of an agent is described by the Arrow-Pratt index \cite{arrow0}, \cite{arrow}, \cite{pratt}: for a utility function $u$ of class $\mathcal{C}^2$, the Arrow-Pratt index is defined by

$r_u(w)=-\frac{u''(w)}{u'(w)}$ for $w \in {\mathbb{R}}$  (3.4)

If $u_1$, $u_2$ are the utility functions of the two agents, then the Arrow-Pratt theorem \cite{arrow0}, \cite{arrow}, \cite{pratt} asserts that "the agent $u_1$ is more risk-averse than the agent $u_2$" iff $r_{u_1}(w) \geq r_{u_2}(w)$ for any $w \in {\mathbb{R}}$.

Papers \cite{georgescu1}, \cite{georgescu2} contain two distinct
possibilistic treatments of  risk aversion when the risk is a fuzzy number. These possibilistic theories of risk aversion are based on the possibilistic utilities $T_1(A, u)$, $T_2(A, u)$. In particular, in both cases a Pratt type theorem is proved. A surprising result is obtained: the possibilistic risk aversion is characterized in terms of the Arrow-Pratt index. In a certain sense, we could say that "the possibilistic risk aversion (in the sense of papers \cite{georgescu1}, \cite{georgescu2}) is equivalent to the probabilistic risk aversion" \cite{arrow0}, \cite{arrow}, \cite{pratt}.

The expected utility operators allow a generalization of possibilistic risk aversion theories from \cite{georgescu1}, \cite{georgescu2}.

In this general framework it is defined what it means that "an agent is more risk-averse than another agent" and it is proved a Pratt-type theorem which characterizes this property. The main tool used in proving this result is the approximation formula from the following proposition.

\begin{propozitie}
\cite{georgescu2}, \cite{georgescu3} Let $T$ be an expected utility operator, $A$ a fuzzy number and $u$ a utility function of class $\mathcal{C}^2$. Then

$T(A, u) \approx u(E_f(A))+\frac{1}{2} u''(E_f(A))Var_T(A)$ (3.5)
\end{propozitie}

Proposition 3.5 will be used in this paper to prove the approximation formula from Theorem 5.9.

Still the treatment of other topics from risk theory in the framework offered  by expected utility operators is not possible without imposing some supplementary  conditions on those. In paper \cite{georgescu6} the $D$-operators have been introduced to study a possibilistic portfolio choice problem. We will present next the definition of the $D$-operators.

For a utility function $g(x, \lambda)$, in which $\lambda$ is a parameter, we consider the following properties:

(i) $g(x, \lambda)$ is continuous with respect to the argument $x$ and derivable with respect to the argument $\lambda$;

(ii) for any $\lambda \in {\mathbb{R}}$, the function $\frac{\partial  g(., \lambda)}{\partial \lambda} :{\mathbb{R}} \rightarrow {\mathbb{R}}$ is derivable.

\begin{definitie}
\cite{georgescu6} An expected utility operator $T: \mathcal{F} \times \mathcal{C}({\mathbb{R}}) \rightarrow {\mathbb{R}}$ is a $D$-\emph{operator} if for any fuzzy number $A$ and for any function $g(x, \lambda)$ with the properties (i) and (ii), the following axioms are fulfilled:

($D_1$) The function $\lambda \mapsto T(A, g(., \lambda))$ is derivable (with respect to $\lambda$);

($D_2$) $T(A, \frac{\partial g(., \lambda)}{\partial \lambda})=\frac{d}{d\lambda }T(A, g(., \lambda))$.
\end{definitie}

By Proposition 1 from \cite{georgescu6}, $T_1$ and $T_2$ are $D$-operators.

\begin{remarca}
Conditions ($D_1$) and ($D_2$) make it possible the use of first order conditions in solving some optimization problems in which the objective function is a $T$-expected utility. In paper \cite{georgescu6}, the $D$-operators offer the framework to find some approximate solutions of a possibilistic portfolio problem. In the next sections, the axioms $(D_1$) and ($D_2$) will be intensely used to determine the optimal coinsurance rate in a coinsurance problem formulated in the context of expected utility operators.
\end{remarca}

\begin{propozitie}
Let $T, S$ be two expected utility operators and $c \in {\mathbb{R}}$.

(a) $U=cT+(1-c)S$ is an expected utility operator;

(b) If $S, T$ are $D$-operators then $U$ is a $D$-operator.
\end{propozitie}

\begin{proof}
(a) By \cite{georgescu3}, Proposition 5.2.5.

(b) The axiom ($D_1$) is immediate. We will verify ($D_2$). Since $S, T$ fulfill condition ($D_2$), we will have the following equalities:

$\frac{d}{d \lambda}U(A, g(., \lambda))=c \frac{d}{d \lambda}T(A, g(., \lambda))+(1-c) \frac{d}{d \lambda}S(A, g(., \lambda))$

\hspace{2cm}$=cT(A, \frac{\partial g(., \lambda)}{\partial \lambda})+(1-c) S(A, \frac{\partial g(., \lambda)}{\partial \lambda})$

\hspace{2cm} $=U(A,  \frac{\partial g(., \lambda)}{\partial \lambda})$

for any fuzzy number $A$ and for any utility function $g(x, \lambda)$ which fulfils hypotheses (i) and (ii).
\end{proof}

\section{The coinsurance problem}

In this section we will deal with the coinsurance problem in the context of expected utility operators. First we will introduce a few entities by which we will define this coinsurance problem, then we will restrict the universe of discussion to $D$-operators in order to start the study of optimal coinsurance rate.

Consider an agent with a utility function $u$ of class $C^2$ such that $u'>0$ and $u''<0$. Assume that the agent possesses an initial wealth subject to risk. To retrieve a part of the loss caused by this risk, the agent will close an insurance contract. By \cite{eeckhoudt}, p. 46, an insurance contract has two components:

$\bullet$  a premium $P$ to be paid by the policyholder;

$\bullet$ an indemnity schedule $I(x)$ indicates the amount to be paid by the insurer for a loss $x$.

We will think of $I(x)$ as a utility function, and the premium $P$ will be defined with respect to the mathematical modeling of risk. In case of a probabilistic model, the loss will be a random variable $X \geq 0$, and $P$ will be defined by means of (probabilistic) expected utility $EI(X)$ (see \cite{eeckhoudt}, p. 49).

The possibilistic form of the coinsurance problem from \cite{georgescu5} has as hypothesis the fact that the risk is a fuzzy number $A$ with the property that $supp(A) \subseteq {\mathbb{R}}_{+}$ and $supp(A)$ does not reduce to a single point. In particular, this hypothesis assures that $E_f(A)>0$.

To extend this to a $EU$-theory associated with an expected utility operator $T$, we will fix $T$ and a weighting function $f: [0, 1] \rightarrow {\mathbb{R}}$.

The \emph{$T$-premium for insurance indemnity} is defined by

$P=(1+\lambda)T(A, u)$ (4.1)

where $\lambda \geq 0$ is a loading factor.

\begin{remarca}
\textnormal
(i) The expression (4.1) of $P$ is inspired from the form of the premium for insurance indemnity from the probabilistic model (\cite{eeckhoudt}, p. 42).

(ii) If $T$ is the operator $T_1$ from  Example 3.2, then we obtain the notion of possibilistic premium for insurance indemnity from \cite{georgescu5}.
\end{remarca}

We will assume that $I(x)=\beta x$ for all $x$. Following the terminology from \cite{eeckhoudt}, $\beta$ will be called coinsurance rate, and $1-\beta$ will be called {\emph{retention rate}}. The coinsurance rate $\beta$ represents the fraction from the size of the loss the insured gets following an insurance contract.

Similar with \cite{eeckhoudt}, p. 49 or \cite{georgescu5}, we will make the hypothesis that the policyholder chooses {\emph{apriori}} a coinsurance rate $\beta$. Then the corresponding $T$-premium for insurance indemnity $P(\beta)$ will have the form:

$P(\beta)=(1+\lambda)T(A, \beta x)=\beta (1+\lambda) T(A, x)$.

By the axiom (a) from Definition 3.1, $P(\beta)$ will be written:

$P(\beta)=\beta(1+\lambda)E_f(A)$ (4.2)

By denoting $P_0=(1+\lambda)E_f(A)$ we will have

$P(\beta)=\beta P_0$ (4.3)

If $\beta$ is the coinsurance rate and $x$ is the size of the loss then the agent remains ultimately with the following amount:

$g(x, \beta)=w_0-P(\beta)-x+\beta x=w_0-\beta P_0-(1-\beta)x$ (4.4)

Consider the function which gives the utility of amount $g(x, \beta)$:

$h(x, \beta)=u(g(x, \beta))=u(w_0-\beta P_0-(1-\beta)x)$ (4.5)

Then

$H(\beta)=T(A, h(x, \beta))$ (4.6)

is the total $T$-utility associated with a possibilistic risk $A$, an initial wealth $w_0$ and an insurance contract with a coinsurance rate $\beta$.

Since the agent wants to maximize this total utility, he will choose $\beta$ as the solution of an optimization problem:

$\displaystyle \max_{\beta} H(\beta)$ (the coinsurance problem)  (4.7)

To be able to study the existence and the computation of the solution of problem (4.7), we will assume from now on that $T$ is a $D$-operator.

Taking into account the axioms $(D_1)$ and $(D_2)$ of Definition 3.5 we will have

$H'(\beta)=\frac{d}{d \beta}T(A, h(., \beta))=T(A, \frac{\partial }{\partial \beta}h(x, \beta)))$

Taking into account that

$\frac{\partial }{\partial \beta}h(x, \beta)=\frac{\partial }{\partial \beta}u(g(x, \beta))=u'(g(x, \beta))\frac{\partial g(x, \beta)}{\partial \beta}=u'(g(x, \beta))(x-P_0)$

the following form of the derivative of $H(\beta)$ follows:

$H'(\beta)=T(A, u'(g(x, \beta))(x-P_0))$  (4.8)

Analogously, we obtain the second derivative:

$H''(\beta)=T(A, u''(g(x, \beta))(x-P_0)^2)$  (4.9)

By hypothesis, $u''(g(x, \beta))<0$. Applying the axiom (d) of Definition 3.1, from (4.9) it follows $H''(\beta)\leq 0$, thus $H$ is a concave function.
Moreover, if the expected utility operator $T$ is strictly increasing then $H$ is strictly concave.
We can consider then the solution $\beta^\ast_T$ of the optimization problem 4.7: $H^\ast(\beta^\ast_T)=\displaystyle \max_{\beta} H(\beta)$. The determination of the optimal coinsurance rate $\beta^\ast$ and the total utility function $H(\beta^\ast)$ is one of the agent's important problems. When it exists, the optimal coinsurance rate $\beta^\ast$ verifies the first order condition $H'(\beta)=0$.  Taking into account (4.8), the first order condition $H'(\beta)=0$ will be written:

$T(A, (x-P_0)u'(g(x, \beta^\ast_T)))=0$ (4.10)

Let us consider the case of $D$-operators $T_1$ and $T_2$. By (3.1) and (3.2), the first order condition (4.10) gets the following form:

$\bullet$ for the $D$-operator $T_1$:

$E_1(f, (A-P_0)u'(g(A, \beta^\ast_{T_1})))=0$ (4.11)

$\bullet$ for the $D$-operator $T_2$:

$E_2(f, (A-P_0)u'(g(A, \beta^\ast_{T_2})))=0$  (4.12)

The coinsurance problem formulated in $EU$-theory associated with the operator $T_1$ has been studied in \cite{georgescu5}. In particular, using the first-order condition (4.11), in \cite{georgescu5} an approximate calculation formula of the optimal coinsurance has been proved.

\section{The properties and the computation of the optimal coinsurance rate}

Let $f: [a, b] \rightarrow {\mathbb{R}}$ be a weighting function, $T$ a $D$-operator, $u: {\mathbb{R}} \rightarrow {\mathbb{R}}$ a utility function
and $A$ a fuzzy number. As in the previous section, we will make the following assumptions on $u$ and $A$:

$\bullet$ $u$ is of class $\mathcal{C}^2$, $u'>0$ and $u''<0$;

$\bullet$ $supp(A) \subseteq {\mathbb{R}}_+$ and $supp(A)$ is not a point set.

According to the second hypothesis, one gets $E_f(A)>0$.

Let $\beta^\ast=\beta^\ast_T$ be the solution of the insurance problem (4.7). We will keep all notations from Section $4$.

\begin{propozitie}
(i) If $\lambda =0$ then $\beta^\ast=1$;

(ii) If $\lambda >0$ then $\beta^\ast<1$.
\end{propozitie}

\begin{proof}
(i) Setting $\beta=1$ in (4.4), we have $g(x, 1)=w_0-P_0$. Then, by Definition 3.1 it follows

$T(A, (x-P_0)u'(g(x, 1)))=T(A, (x-P_0)u'(w_0-P_0))$

$\hspace{2cm}$ $=u'(w_0-P_0)T(A, x-P_0)$

$\hspace{2cm}$ $=u'(w_0-P_0)(T(A, x)-P_0)$

$\hspace{2cm}$ $=u'(w_0-P_0)(E_f(A)-P_0)$

$\hspace{2cm}$ $=-\lambda E_f(A) u'(w_0-P_0)$

If $\lambda= 0$ then $\beta^\ast=1$ verifies the first order condition (4.10).

(ii) Since $supp(A) \subseteq {\mathbb{R}}_+$ and $supp(A)$ is not a point set, it follows

$E_f(A)=\frac{1}{2} \int_0^1 [a_1(\gamma)+a_2(\gamma)]f(\gamma)d\gamma >0$

From $\lambda >0$, $u'>0$ and the proof of (i) we deduce

$H'(1)=T(A, (x-P_0)u'(g(x, 1)))=-\lambda E_f(A)u'(w_0-P_0)<0$

Suppose by absurdum that $\beta^\ast \geq 1$. Since $H$ is concave, its derivative $H'$ is decreasing, thus $H'(\beta) \leq H'(1)<0$. This contradicts the first order condition $H'(\beta^\ast)=0$, thus $\beta^\ast <1$.

\end{proof}

By Proposition 5.1, the inequality $\beta^\ast \leq 1$ holds. In the formulation (4.7) of the coinsurance problem no restriction has been made on $\beta$ (as in \cite{eeckhoudt}, Section 3.2, for the probabilistic coinsurance rate). Therefore, the solution $\beta^\ast$ of (4.7) may not satisfy the inequality $0< \beta^\ast \leq 1$. In the Appendix, we will establish a necessary condition for $0< \beta^\ast$.

\begin{remarca}
Proposition 5.1 is a result analogous to Mosin theorem (\cite{mossin} or \cite{eeckhoudt}, Proposition 3.1). Then when $T$ is the operator $T_1$ one obtains Proposition 4 from \cite{georgescu5}.
\end{remarca}

An exact solution for the maximization problem (4.7) is difficult to find. Therefore, it is more convenient to find approximate solutions of equation (4.10).

Before proving a formula for the approximate calculation of $\beta^\ast$, let us denote $w=w_0-P_0$. Then formula (4.4) becomes:

$g(x, \beta)=w-(1-\beta)(x-P_0)$  (5.1)

\begin{teorema}
An approximate value of the optimal $T$-coinsurance rate $\beta^\ast$ is:

$\beta^\ast \approx 1+\frac{u'(w)}{u''(w)}\frac{\lambda E_f(A)} {Var_T(A)+\lambda^2 E_f^2(A)}$
\end{teorema}

\begin{proof}
By (5.1), $u'(g(x, \beta))=u'(w-(1-\beta)(x-P_0))$. We consider the first-order Taylor approximation of $u'(w-(1-\beta)(x-P_0))$ around $w$:

$u'(g(x, \beta)) \approx u'(w)-(1-\beta)(x-P_0)u''(w)$

from where it follows

$(x-P_0)u'(g(x, \beta))\approx u'(w) (x-P_0)-(1-\beta)u''(w)(x-P_0)^2$

Taking into account (4.8) and Definition 3.1, we have

$H'(\beta)=T(A, (x-P_0)u'(g(x, \beta))$

$\hspace{1.1cm}$ $\approx T(A, u'(w)(x-P_0)-(1-\beta)u''(w)(x-P_0)^2)$

$\hspace{1.1cm}$ $=u'(w)T(A, x-P_0)-(1-\beta)u''(w)T(A, (x-P_0)^2)$

We notice that

$T(A, x-P_0)=T(A, x)-P_0=E_f(A)-P_0$.

$T(A, (x-P_0)^2)=T(A, x^2-2P_0x+P_0^2)$

$\hspace{2.3cm}$ $=T(A, x^2)-2P_0T(A, x)+P_0^2$

$\hspace{2.3cm}$ $=T(A, x^2)-2P_0E_f(A)+P_0^2$

$\hspace{2.3cm}$ $=(T(A, x^2)-E_f^2(A))+(E_f(A)-P_0)^2$

$\hspace{2.3cm}$ $=Var_T(A)+(E_f(A)-P_0)^2$

Replacing $T(A, x-P_0)$ and $T(A, (x-P_0)^2)$ in the approximate expression of $H'(\beta)$ one obtains:

$H'(\beta) \approx u'(w) (E_f(A)-P_0)-(1-\beta)u''(w) [Var_T(A)+(E_f(A)-P_0)^2]$

Then the first-order condition $H'(\beta^\ast)=0$ can be written

$u'(w)(E_f(A)-P_0)-(1-\beta^\ast)u''(w)[Var_T(A)+(E_f(A)-P_0)^2] \approx 0$

from where

$\beta^\ast \approx 1-\frac{u'(w)}{u''(w)}\frac{E_f(A)-P_0}{Var_T(A)+(E_f(A)-P_0)^2}$

Since $P_0=(1+\lambda)E_f(A)$, we have $E_f(A)-P_0=E_f(A)-(1+\lambda)E_f(A)=-\lambda E_f(A)$. With this, the approximate value of $\beta^\ast$ gets the form

$\beta^\ast \approx 1+\frac{u'(w)}{u''(w)}\frac{\lambda E_f(A)}{Var_T(A)+\lambda^2 E_f^2(A)}$
\end{proof}

Taking into account the definition of the Arrow-Pratt index from (3.4) one obtains

\begin{corolar}
$\beta^\ast \approx 1-\frac{\lambda }{r_u(w)}\frac{E_f(A)}{Var_T(A)+\lambda^2 E_f^2(A)}$  (5.3)
\end{corolar}

By particularizing the operator $T$, different approximation formulas of the optimal coinsurance rate $\beta$ are obtained from (5.3). If $T$ is the $D$-operator $T_1$ from Example 3.2, then the approximation formula (22) from \cite{georgescu5} is obtained.

\begin{remarca}
The approximate value of $\beta^\ast$ given by (5.3) gives us the way the optimal $T$-insurance depends on the risk aversion of the agent who closes the insurance contract, as well as it depends on the expected value and the variance of the fuzzy number representing the risk. The following result will give a more precise form of the relation between the risk aversion and the $T$-coinsurance rate: an increase in risk aversion will generate an increase in coinsurance rate.
\end{remarca}

We consider two agents whose utility functions $u_1$, $u_2$ are of class $C^2$ and verify the conditions $u'_1>0$, $u'_2>0$, $u''_1 <0$ and $u''_2<0$. Let $\beta_1^\ast$, $\beta_2^\ast$ be the optimal $T$-coinsurance rates associated with the utility functions $u_1$, $u_2$, the weighting function $f$, the $D$-operator $T$ and the fuzzy number $A$.

\begin{propozitie}
If the agent $u_1$ is more risk-averse than $u_2$, then $\beta_1^\ast >\beta_2^\ast$.
\end{propozitie}

\begin{proof}
We consider the Arrow-Pratt indices of the utility functions $u_1$, $u_2$:

$r_{u_1}(w)=-\frac{u''_1(w)}{u'_1(w)}$, $r_{u_2}(w)=-\frac{u''_2(w)}{u'_2(w)}$

By $u'_1>0$, $u'_2>0$, $u''_1 <0$ and $u''_2<0$ we have $r_{u_1}(w)>0$ and $r_{u_2}(w)>0$ for any $w$. We apply Corollary 5.4 in case of the coinsurance problem corresponding to the agents $u_1$ and $u_2$:

$\beta^\ast_1 \approx 1-\frac{\lambda}{r_{u_1}(w)}\frac{E_f(A)}{Var_T(A)+\lambda^2 E_f^2(A)}$  (5.4)

$\beta^\ast_2 \approx 1-\frac{\lambda}{r_{u_2}(w)}\frac{E_f(A)}{Var_T(A)+\lambda^2 E_f^2(A)}$  (5.5)

By hypothesis, $r_{u_1}(w) \geq r_{u_2}(w)>0$, thus $0<\frac{\lambda}{r_{u_1}(w)}<\frac{\lambda}{r_{u_2}(w)}$. Since $E_f(A)>0$ and $Var_T(A) \geq 0$, from (5.4) and (5.5) it follows immediately $\beta^\ast_1 \geq \beta^\ast_2$.
\end{proof}

Let $T, S$ be two $T$-operators and $c \in {\mathbb{R}}$. By Proposition 3.8, $U=cT+(1-c)S$ is a $D$-operator. We consider the coinsurance problems associated with the $D$-operators $T, S, U$ and in rest, keeping the same data which define the coinsurance problem (4.7).

\begin{propozitie}
Let $\beta^\ast_T$, $\beta^\ast_S$ and $\beta^\ast_U$ the optimal coinsurance rates corresponding to the $D$-operators $T, S, U$. Then

$\beta^\ast_U \approx 1-\frac{1}{\frac{c}{1-\beta^\ast_T}+\frac{1-c}{1-\beta^\ast_S}}$ (5.6)
\end{propozitie}

\begin{proof}
By (5.3) we have the following approximate values of $\beta^\ast_T$, $\beta^\ast_S$ and $\beta^\ast_U$:

$\beta^\ast_T \approx 1 -\frac{\lambda}{r_u(w)} \frac{E_f(A)}{Var_T(A)+\lambda^2 E_f^2(A)}$

$\beta^\ast_S \approx 1 -\frac{\lambda}{r_u(w)} \frac{E_f(A)}{Var_S(A)+\lambda^2 E_f^2(A)}$

$\beta^\ast_U \approx 1 -\frac{\lambda}{r_u(w)} \frac{E_f(A)}{Var_U(A)+\lambda^2 E_f^2(A)}$

from where it follows:

$\frac{1}{1-\beta^\ast_T} \approx \frac{r_u(w)}{\lambda}\frac{Var_T(A)+\lambda^2 E_f^2(A)}{E_f(A)}$  (5.7)

$\frac{1}{1-\beta^\ast_S} \approx \frac{r_u(w)}{\lambda}\frac{Var_S(A)+\lambda^2 E_f^2(A)}{E_f(A)}$   (5.8)

$\frac{1}{1-\beta^\ast_U}\approx \frac{r_u(w)}{\lambda}\frac{Var_U(A)+\lambda^2 E_f^2(A)}{E_f(A)}$   (5.9)

By \cite{georgescu3}, Proposition 5.1.5, we have $Var_U(A)=cVar_T(A)+(1-c)Var_S(A)$. Then, by taking into account (5.7)-(5.9) the following equalities hold:

$\frac{c}{1-\beta^\ast_T}+\frac{1-c}{1-\beta^\ast_S} \approx \frac{r_u(w)}{\lambda} \frac{cVar_T(A)+(1-c)Var_S(A)+\lambda^2E_f^2(A)}{E_f(A)}$

$\hspace{2cm}$ $\approx \frac{r_u(w)}{\lambda} \frac{Var_U(A)+\lambda^2 E_f^2(A)}{E_f(A)}$

$\hspace{2cm}$ $\approx \frac{1}{1-\beta^\ast_U}$

From here it follows

$\beta_U^\ast \approx 1-\frac{1}{\frac{c}{1-\beta^\ast_T}+\frac{1-c}{1-\beta^\ast_S}}$.

\end{proof}

\begin{remarca}
The previous proposition allows to obtain the optimal coinsurance rates for all convex combinations of two $D$-operators. In particular, if we take $c=\frac{1}{2}$ then $U=\frac{1}{2}T+\frac{1}{2}S$ and by (5.6), the optimal coinsurance rate of $U$ will be:

$\beta^\ast_U \approx 1-\frac{2}{\frac{1}{1-\beta^\ast_T}+\frac{1}{1-\beta^\ast_S}}$ (5.10)
\end{remarca}

The approximation formula of the optimal insurance $\beta^\ast$ from Corollary 5.4 will be used in the following to approximate the total expected utility $H(\beta^\ast)=T(A, h(x, \beta^\ast))$.

\begin{teorema}
The total expected utility $H(\beta^\ast)$ corresponding to the optimal coinsurance rate $\beta^\ast$ can be approximated by

$H(\beta^\ast) \approx u(w+\frac{1}{r_u(w)}\frac{\lambda^2 E_f^2(A)}{Var_T(A)+\lambda^2 E_f^2(A)})+$

$\frac{\lambda^2}{2 r_u^2(w)}\frac{E_f^2(A) Var_T(A}{(Var_T(A)+\lambda^2 E_f^2(A))^2}u''(w+\frac{1}{r_u(w)}\frac{\lambda^2 E_f^2(A)}{Var_T(A)+\lambda^2E_f^2(A)})$
\end{teorema}

\begin{proof}
We consider the unidimensional function

$v(x)=h(x, \beta^\ast)=u(w-(1-\beta^\ast)(x-P_0))$  (5.11)

One will notice that $H(\beta^\ast)=T(A, v)$. $v$ is a utility function of class $\mathcal{C}^2$ thus we can apply the approximation formula (3.5):

$H(\beta^\ast) \approx v(E_f(A))+\frac{v''(E_f(A))}{2}Var_T(A)$ (5.12)

Since $E_f(A)-P_0=-\lambda E_f(A)$ it follows the approximation

$w-(1-\beta^\ast)(E_f(A)-P_0)=w+\lambda (1-\beta^\ast)E_f(A)$  (5.13)

By Corollary 5.4 one has $1-\beta^\ast \approx \frac{\lambda}{r_u(w)} \frac{ E_f(A)}{Var_T(A)+\lambda^2 E_f^2(A)}$ (5.14)

Deriving twice in (5.11) it follows

$v''(x)=(1-\beta^\ast)^2 u''(w-(1-\beta^\ast{^2})(x-P_0))$ (5.15)

From (5.11), (5.15) and (5.14) we can deduce

$v(E_f(A)) \approx u(w+\frac{1}{r_u(w)}
\frac{\lambda^2 E_f^2(A)}{Var_T(A)+\lambda^2 E_f^2(A)})$  (5.16)

$v''(E_f(A)) \approx \frac{\lambda^2 }{r_u^2(w)}\frac{E_f^2(A)}{(Var_T(A)+\lambda^2E_f^2(A))^2}u''(w+\frac{1}{r_u(w)}\frac{\lambda^2 E_f^2(A)}{Var_T(A)+\lambda^2 E_f^2(A)})$  (5.17)

Replacing $v(E_f(A))$ and $v''(E_f(A))$ with their approximate values from (5.14) and (5.17) the approximation formula from the statement of the theorem follows.
\end{proof}

\section{Particular cases and examples}

In this section we will study the optimal $T$-coinsurance rate for some particular $D$-operators, making the following assumptions on the weighting function $f$ and the fuzzy number $A$:

$\bullet$ $f(t)=2t$, for any $t \in [0, 1]$;

$\bullet$ $A$ is the triangular fuzzy number $(a, \alpha, \beta)$:

$$
A(t) = \left\{
        \begin{array}{ll}
            1-\frac{a-t}{\alpha} & \quad a-\alpha \leq t \leq a \\
            1- \frac{t-a}{\beta}   & \quad a \leq t \leq a+ \beta\\
            0    & \quad otherwise
        \end{array}
    \right.
$$

As to the $D$-operator $T$, we will consider the following particular cases:

(a) \emph{$T$ is the $D$-operator $T_1$.} By \cite{georgescu3}, Examples 3.3.10 and 3.4.10 we have

$E_f(A)=a+\frac{\beta-\alpha}{6}$ (6.1)

$Var_{T_1}(A)=\frac{\alpha^2+\beta^2+\alpha \beta}{18}$ (6.2)

Replacing $E_f(A)$ and $Var_{T_1}(A)$ in (5.3), the optimal $T_1$-coinsurance rate $\beta^\ast_1=\beta^\ast_{T_1}$ gets the form:

$\beta^\ast_1 \approx 1- \frac{\lambda }{r_u(w)} \frac{a+\frac{\beta-\alpha}{6}}{\frac{\alpha^2+\beta^2+\alpha \beta}{18}+\lambda^2
(a+\frac{\beta-\alpha}{6})^2}$  (6.3)

(b) \emph{$T$ is the $D$-operator $T_2$.} By \cite{georgescu3}, Example 3.4.10 we have

$Var_{T_2}(A)=\frac{\alpha^2+\beta^2}{36}$ (6.4)

Replacing $E_f(A)$ and $Var_{T_2}(A)$ in (5.3), the optimal $T_2$-coinsurance rate $\beta^\ast_{T_2}$ becomes:

$\beta^\ast_2 \approx 1- \frac{\lambda }{r_u(w)} \frac{a+\frac{\beta-\alpha}{6}}{\frac{\alpha^2+\beta^2}{36}+\lambda^2
(a+\frac{\beta-\alpha}{6})^2}$  (6.5)

(c) We consider the $D$-operator $U=\frac{1}{2}T_1+\frac{1}{2}T_2$  (by Proposition 3.8). For the computation of the optimal $U$-coinsurance rate $\beta^\ast_u$ we will  recall the formula from Remark 5.8. Using (6.3) and (6.5) one obtains:

$\frac{1}{1-\beta^\ast_1}+\frac{1}{1-\beta^\ast_2} \approx \frac{r_u(w)}{\lambda} \frac{\frac{\alpha^2+\beta^2+\alpha \beta}{18}+\frac{\alpha^2+\beta^2}{36}+2 \lambda^2(a+\frac{\beta-\alpha}{6})^2}{a+\frac{\beta-\alpha}{6}}$

$=\frac{r_u(w)}{\lambda} \frac{\frac{(\alpha+\beta)^2+2(\alpha^2+\beta^2)}{36}+2 \lambda^2(a+\frac{\beta-\alpha}{6})^2}{a+\frac{\beta-\alpha}{6}}$

By Remark 5.8, one gets

$\beta^\ast_U=1-\frac{2\lambda}{r_u(w)}\frac{a+\frac{\beta-\alpha}{6}}{\frac{(\alpha+\beta)^2+2(\alpha^2+\beta^2)}{36}+2 \lambda^2(a+\frac{\beta-\alpha}{6})^2}$.  (6.6)

We ask the problem of comparing the two coinsurance rates $\beta^\ast_1$, $\beta^\ast_2$ from (6.3) and (6.5). First we notice that if $\lambda=0$, then by Proposition 5.1 (i), we have $\beta^\ast_1=\beta^\ast_2=1$.

\begin{propozitie}
If $\lambda>0$ then there is the following dependence relation between $\beta^\ast_1$ and $\beta^\ast_2$:

$\frac{1}{1-\beta^\ast_1}-\frac{1}{1-\beta^\ast_2} \approx \frac{(\alpha+\beta)^2}{36\lambda E_f(A)} r_u(w)$ (6.7)
\end{propozitie}

\begin{proof}
Formulas (6.3) and (6.5) can be written

$\frac{\alpha^2+\beta^2+\alpha \beta}{18}+\lambda^2 E_f^2(A) \approx \frac{\lambda E_f(A)}{r_u(w)(1-\beta^\ast_1)}$

$\frac{\alpha^2+\beta^2}{36}+\lambda^2 E_f^2(A) \approx \frac{\lambda E_f(A)}{r_u(w)(1-\beta^\ast_2)}$

(By Proposition 5.1 (ii), $1-\beta^\ast_1>0$ and $1-\beta^\ast_2>0$)

By subtotal, from the two previous relations it follows:

$\frac{(\alpha+\beta)^2}{36} \approx \frac{\lambda E_f(A)}{r_u(w)}[\frac{1}{1-\beta^\ast_1}-\frac{1}{1-\beta^\ast_2} ]$

which implies (6.7).
\end{proof}

\begin{corolar}
If $\lambda >0$ then $\beta^\ast_1
 > \beta^\ast_2$.
\end{corolar}

\begin{proof}
Since $u'(w)>0$, $u''(w)<0$ implies $r_u(w)=-\frac{u''(w)}{u'(w)}>0$. We have $\lambda >0$ and $E_f(A)>0$, therefore the right hand side of (6.7) is positive. Further, using (6.7) one obtains the inequality $\frac{1}{1-\beta^\ast_1}> \frac{1}{1-\beta^\ast_2}$, from where it follows $\beta^\ast_1
 > \beta^\ast_2$.
\end{proof}

Formulas (6.3), (6.5) and (6.6) may get different forms with respect to the utility function $u$.

\begin{exemplu}
Assume that the utility function $u$ is $HARA$-type (\cite{gollier}, Section 3.6)

$u(w)=\zeta(\eta+\frac{w}{\gamma})^{1-\gamma}$, for $\eta+ \frac{w}{\gamma}>0$ (6.8)

By \cite{gollier}, Section 3.6, $r_u(w)=(\eta+\frac{w}{\gamma})^{-1}$, thus formulas (6.3), (6.5) and (6.6) will get the form:

$\beta^\ast_1 \approx 1-\lambda (\eta+\frac{w}{\gamma})\frac{a+\frac{\beta-\alpha}{6}}{\frac{\alpha^2+\beta^2+\alpha\beta}{18}+\lambda^2(a+\frac{\beta-\alpha}{6})^2}$  (6.9)

$\beta^\ast_2 \approx 1-\lambda (\eta+\frac{w}{\gamma})\frac{a+\frac{\beta-\alpha}{6}}{\frac{\alpha^2+\beta^2}{36}+\lambda^2(a+\frac{\beta-\alpha}{6})^2}$ (6.10)

$\beta^\ast_U \approx 1-2\lambda (\eta+\frac{w}{\gamma})\frac{a+\frac{\beta-\alpha}{6}}{\frac{(\alpha+\beta)^2+2(\alpha^2+\beta^2)}{36}+2\lambda^2(a+\frac{\beta-\alpha}{6})^2}$  (6.11)

If $A$ is a symmetric triangular fuzzy number $(a, \alpha)$, then, setting $\beta=\alpha$ in (6.9)-(6.11) we find the following forms of the three optimal coinsurance rates:

$\beta^\ast_1 \approx 1-\lambda (\eta+\frac{w}{\gamma})\frac{6a}{\alpha^2+6\lambda^2a^2}$   (6.12)

$\beta^\ast_2 \approx 1-\lambda (\eta+\frac{w}{\gamma})\frac{18a}{\alpha^2+18\lambda^2a^2}$   (6.13)

$\beta^\ast_U \approx 1-\lambda (\eta+\frac{w}{\gamma})\frac{9a}{2\alpha^2+18\lambda^2a^2}$   (6.14)

\end{exemplu}

\begin{exemplu}
Assume that the utility function $u$ is CRRA-type:

$$
u(w) = \left\{
        \begin{array}{ll}
            \frac{w^{1-\gamma}}{1-\gamma} & \quad \gamma > 1 \\
            ln(w)   & \quad \gamma=1
        \end{array}   (6.15)
    \right.
$$

Then, by \cite{eeckhoudt}, p. 21, $r_u(w)=\frac{\gamma}{w}$ for $\gamma>1$ and $r_u(w)=\frac{1}{w}$ for $\gamma=1$. Then, by (6.3), (6.5) and (6.6) the following formulas for the optimal coinsurance rates $\beta^\ast_1$, $\beta^\ast_2$, $\beta^\ast_U$ follow:

$\bullet$ for $\gamma>1$:

$\beta^\ast_1 \approx 1-\frac{\lambda w}{\gamma} \frac{a+\frac{\beta-\alpha}{6}}{\frac{\alpha^2+\beta^2+\alpha\beta}{18}+\lambda^2(a+\frac{\beta-\alpha}{6})^2}$  (6.16)

$\beta^\ast_2 \approx 1-\frac{\lambda w}{\gamma} \frac{a+\frac{\beta-\alpha}{6}}{\frac{\alpha^2+\beta^2}{36}+\lambda^2(a+\frac{\beta-\alpha}{6})^2}$ (6.17)

$\beta^\ast_U \approx 1-2\frac{\lambda w}{\gamma} \frac{a+\frac{\beta-\alpha}{6}}{\frac{(\alpha+\beta)^2+2(\alpha^2+\beta^2)}{36}+2\lambda^2(a+\frac{\beta-\alpha}{6})^2}$  (6.18)

$\bullet$ for $\gamma=1$:

$\beta^\ast_1 \approx 1-\lambda w \frac{a+\frac{\beta-\alpha}{6}}{\frac{\alpha^2+\beta^2+\alpha\beta}{18}+\lambda^2(a+\frac{\beta-\alpha}{6})^2}$  (6.19)

$\beta^\ast_2 \approx 1-\lambda w \frac{a+\frac{\beta-\alpha}{6}}{\frac{\alpha^2+\beta^2}{36}+\lambda^2(a+\frac{\beta-\alpha}{6})^2}$ (6.20)

$\beta^\ast_U \approx 1-2 \lambda w \frac{a+\frac{\beta-\alpha}{6}}{\frac{(\alpha+\beta)^2+2(\alpha^2+\beta^2)}{36}+2\lambda^2(a+\frac{\beta-\alpha}{6})^2}$  (6.21)
\end{exemplu}

\begin{exemplu}
We consider the $T$-coinsurance problem with the following initial data:

$\bullet$ the weighting function is $f(t)=2t$, $t \in [0, 1]$;

$\bullet$ $A$ is the triangular fuzzy number $A=(6,2,3)$;

$\bullet$ $u$ is the utility function of CRRA-type $u(w)=ln(w)$, thus $r_u(w)=1$;

$\bullet$ the initial wealth is $w_0=40$ and the loading factor is $\lambda=\frac{1}{2}$.

Formulas (6.1)-(6.3) give the following indicators of the fuzzy number $A$:

$E_f(A)=\frac{37}{6}$, $Var_{T_1}(A)=\frac{19}{18}$, $Var_{T_2}(A)=\frac{13}{36}$

By (4.1), $P_0=(1+\lambda)E_f(A)=\frac{37}{4}$, thus $w=w_0-P_0=\frac{123}{4}$.

Applying formulas (6.3), (6.5) or (6.7), the two optimal coinsurance rates have the approximate values:

$\beta^\ast_1 \approx 1-\lambda w \frac{E_f(A)}{Var_{T_1}(A)+\lambda^2E_f^2(A)}=-10.71$

$\beta^\ast_2 \approx 1-\lambda w \frac{E_f(A)}{Var_{T_2}(A)+\lambda^2E_f^2(A)}=-11.5$
\end{exemplu}

The example above emphasized two $T$-coinsurance problems in which the $T$-coinsurance rated have been strictly negative. In the Appendix we will find a necessary condition for the $T$-coinsurance rate $\beta^\ast_T$ to be strictly positive. We do not know a necessary and sufficient condition for $\beta^\ast_T >0$. For a particular case of the utility function, the following property will give us a sufficient condition for $\beta^\ast_T >0$.

\begin{propozitie}
Assume that the utility function $u$ is defined by: $u(x)=-e^{-x}$, for $x \in {\bf{R}}$. If $\lambda > \frac{1}{E_f(A)}$ then $\beta^\ast_T >0$.
\end{propozitie}

\begin{proof}
One notices immediately that $E_f(A)>0$, thus $\lambda >0$. Also, $r_u(x)=-\frac{u''(x)}{u'(x)}=1$ for any  $x \in {\bf{R}}$, thus, according to (5.3), the optimal $T$-coinsurance rate $\beta^\ast=\beta^\ast_T$ can be approximated as:

$\beta^\ast \approx 1-\frac{\lambda E_f(A)}{Var_T(A)+\lambda^2 E_f^2(A)}$  (6.22)

By hypothesis, $\lambda >\frac{1}{E_f(A)}$, we will have

$E_f^2(A)\lambda ^2 -\lambda E_f(A)+Var_T(A)=(E_f(A)\lambda -1)^2 +\lambda E_f(A)-1+Var_T(A)>0$.

Dividing both members of the previous inequality by $Var_T(A)+\lambda^2 E_f^2(A)>0$ and taking into account (6.22) it follows $\beta^\ast>0$.
\end{proof}

\begin{exemplu}
We consider the following hypotheses:

$\bullet$ the weighting function is $f(t)=2t$, $t \in [0,1]$;

$\bullet$ the risk is represented by the triangular fuzzy number $A=(2, 4, 1)$;

$\bullet$ the utility function is $u(x)=-e^{-x}$, for $x \in {\bf{R}}$;

$\bullet$ the loading factor is $\lambda >0$.

Using the formulas (6.1)-(6.3) we obtain the following indicators:

$E_f(A)=\frac{3}{2}$; $Var_{T_1}(A)=\frac{7}{6}$; $Var_{T_2}(A)=\frac{17}{36}$  (6.23)

Then, by applying the approximation (6.22) of $\beta^\ast_T$ in case of $D$-operators $T_1$, $T_2$ we find:

$\beta^\ast_{T_1} \approx 1-\frac{18 \lambda}{14+27 \lambda^2}$;

$\beta^\ast_{T_2} \approx 1-\frac{54 \lambda}{17+ 81 \lambda^2}$.

In this case the condition of Proposition 6.6 is $\lambda >\frac{2}{3}$. In particular for $\lambda =1$ we obtain $\beta^\ast_{T_1} =\frac{23}{41}$, $\beta^\ast_{T_2}= \frac{22}{49}$.
\end{exemplu}

\section{Concluding Remarks}

The  basic idea of the paper is the study of the coinsurance problem by the expected utility operators from \cite{georgescu2}, \cite{georgescu3}. The main contributions of the paper are:

$\bullet$ to build a coinsurance model in the framework offered by the possibilistic $EU$-theory associated with an expected utility operator;

$\bullet$ the use of $D$-operators defined in \cite{georgescu6} to study the properties of the optimal coinsurance and its approximate calculation;

$\bullet$ the application of the general results to the computation of the coinsurance rates in a few remarkable cases and their comparison.

We report next a few open problems:

(a) We assume that we have a data set representing values of a probabilistic risk (random variables) which appears as a parameter in the context of a probabilistic model (for example, in the coinsurance problem). Based on the existing data, one could determine those indicators by which we know the phenomenon described by the probabilistic model. Thereby, in case of the coinsurance problem, from data one obtains the statistic mean value and variance, then we can compute the optimal coinsurance (by an approximate calculation formula similar to (5.3)). In \cite{vercher}, Vercher et al. present a method by which from a dataset one can build a trapezoidal fuzzy number. By applying Vercher et al.'s method, the probabilistic model of the coinsurance turns into a possibilistic model, in which risk is described by this trapezoidal fuzzy number. We compute then the expected value and the variance associated with this trapezoidal fuzzy number, then by formula (5.3) on can obtain the optimal coinsurance associated with the $T$-possibilistic model. An open problem is to find those formulas describing the way the probabilistic model of coinsurance is turned into a possibilistic model (by Vercher et al's method), which allows a comparison of the two models.

(b) In papers \cite{athayde}, \cite{garlappi}, \cite{niguez} it is studied the effect of absolute risk aversion, prudence and temperance on the optimal solution  for the standard portfolio choice problem (\cite{eeckhoudt}, Section 4.1). A similar problem is investigated in \cite{georgescu6} in the context of $EU$-theory associated with a $D$-operator. It would be interested to study refinements of Theorems 5.3 and 5.9 such that the optimal $T$-coinsurance rate and the standard expected utility to be expressed according to the indicators of risk aversion, prudence and temperance as well as the $T$-moments of the possibilistic risk represented by the fuzzy number $A$.

(c) A third problem is the study of a coinsurance problem with two types of risk: besides the investment risk a background risk might appear. Both the investment risk and the background risk can be probabilistic (random variables) or possibilistic (fuzzy numbers). Besides the purely probabilistic coinsurance model in which both risks are random variables we have:

$\bullet$ the possibilistic model, in which risks are fuzzy;

$\bullet$ two mixed models, in which a risk is a fuzzy number, and the other is a random variable.

To define such bidimensional coinsurance models it is necessary for the notions of multidimensional possibilistic expected utility (\cite{georgescu3}, Section 6.1) and the mixed expected utility (\cite{georgescu3}, Section 7.1) to be generalized for some "multidimensional expected utility operators".

\section{Appendix}

In the following we will prove a necessary condition for the optimal $T$-coinsurance $\beta^\ast=\beta^\ast_T$ to be strictly positive. We will keep the
notations from Sections $4$ and $5$.

\begin{lema}
If $T$ is an expected utility operator, $A$ a fuzzy number and $u, v$ two continuous utility functions then

$T(A, [u(x)-T(A, u(x))][v(x)-T(A, v(x))])=T(A, u(x)v(x))-T(A, u(x))T(A, v(x))$
\end{lema}

\begin{proof}
One uses axioms (b) and (c) from Definition 3.1.
\end{proof}

\begin{propozitie}
Let $T$ be a strictly increasing expected utility operator. Assume that $\lambda >0$. Then from  $\beta^\ast_T >0$ the following inequality follows:

$\lambda < \frac{T(A, (x-E_f(A))[u'(w_0-x)-T(A, u'(w_0-x))])}{E_f(A)T(A, u'(w_0-x))}$
\end{propozitie}

\begin{proof}
We will denote $\beta^\ast=\beta^\ast_T$. From (4.4) we have $g(x, 0)=w_0-x$, thus, by (4.8):

$H'(0)=T(A, u'(w_0-x)(x-P_0))$  (a)

Since $T$ is strictly increasing, $H(\beta)$ is a strictly concave function, thus $H'(\beta)$ is a strictly decreasing function. Then the following implication holds:

$\beta^\ast >0 \Rightarrow 0=H'(\beta^\ast)<H'(0)$  (b)

Applying (a) and Lemma 8.1 it follows

$H'(0)=T(A, [u'(w_0-x)-T(A, u(w_0-x))][x-P_0-T(A, x-P_0)])+T(A, u'(w_0-x))T(A, x-P_0)$

Noticing that $T(A, x-P_0)=-\lambda E_f(A)$, the previous inequality gets the form

$H'(0)=T(A, (x-E_f(A))[u'(w_0-x)-T(A, u'(w_0-x))])-\lambda E_f(A)T(A, u'(w_0-x))$

Then the inequality $H'(0)>0$ is written as:

$T(A, (x-E_f(A))[u'(w_0-x)-T(A, u'(w_0-x_0))])> \lambda E_f(A)T(A, u'(w_0-x))$.

Since $T$ is strictly increasing and $u'(w_0-x)>0$, we have $T(A, u'(w_0-x))>0$. Also $E_f(A)>0$, thus the last inequality from above is equivalent with

$\lambda < \frac{T(A, (x-E_f(A))[u'(w_0-x)-T(A, u'(w_0-x))]}{E_f(A)T(A, u'(w_0-x))}$  (c)

From (b) and (c) the implication which we had to prove follows.

\end{proof}

% that's all folks
\end{document}